\documentclass[11pt]{article}
\usepackage{mathrsfs}
\usepackage{graphicx}
\usepackage{graphics,amsmath,amssymb,amsthm,accents,amsfonts}
\usepackage{epic}
\usepackage{color}
\usepackage{cite}
\usepackage{texdraw}
\usepackage[all]{xy}
\makeatletter

\newcommand{\Rmnum}[1]{\expandafter\@slowromancap\romannumeral #1@}
\makeatother \textwidth=16.7truecm \textheight=22truecm
\newtheorem{lemma}{Lemma}
\newtheorem{Theorem}{Theorem}

\def \wh#1{\widehat{#1}}
\def \wt#1{\widetilde{#1}}

\def \dh#1{\underaccent{\hat}{#1}}

\def\dwh{\underaccent{{\cc@style\widehat{\mskip10mu}}}}
\def \dt#1{\underaccent{\tilde}{#1}}

\makeatother \numberwithin{equation}{section}
\newtheorem{prop}{Proposition}
\allowdisplaybreaks

\newcommand{\nn}{\nonumber}

\newcommand{\ga}{\gamma}
\newcommand{\tta}{\theta}
\newcommand{\vta}{\vartheta}
\newcommand{\del}{\delta}
\newcommand{\al}{\alpha}
\newcommand{\be}{\beta}

\newcommand{\vrr}{\varrho}

\newcommand{\Ph}{\boldsymbol{\phi}}
\newcommand{\Ps}{\boldsymbol{\psi}}
\newcommand{\Ch}{\boldsymbol{\chi}}
\newcommand{\Vph}{\boldsymbol{\varphi}}
\newcommand{\Vap}{\boldsymbol{\varpi}}

\newcommand{\bA}{\boldsymbol{A}}
\newcommand{\bB}{\boldsymbol{B}}
\newcommand{\bC}{\boldsymbol{C}}
\newcommand{\bD}{\boldsymbol{D}}

\newcommand{\UU}{Z}
\newcommand{\UM}{S}

\def \cd#1{\accentset{\circ}{#1}}
\def \dc#1{\underaccent{\circ}{#1}}
\def \d#1{\accentset{\bullet}{#1}}
\def \dd#1{\underaccent{\bullet}{#1}}

\begin{document}

\title{Rational solutions to the ABS list: Degenerating approach}

\author{Song-lin Zhao$^{1}$\footnote{E-mail: songlinzhao@zjut.edu.cn}~~,~~Da-jun Zhang$^{2}$\footnote{E-mail: djzhang@staff.shu.edu.cn}\\
{\small\it $^{1}$ Department of Applied Mathematics, Zhejiang University of Technology, Hangzhou 310023, P.R. China}\\
{\small \it $^{2}$ Department of Mathematics, Shanghai University, Shanghai 200444, P.R. China}}

\maketitle

\begin{abstract}
In the paper we first construct rational solutions for the Nijhoff-Quispel-Capel (NQC) equation by means of bilinear method.
These solutions can be transferred to those of Q3$_\delta$ equation in the Adler-Bobenko-Suris (ABS) list.
Then making use of degeneration relation we obtain rational solutions for
Q2, Q1$_\delta$, H3$_\delta$, H2 and H1.
These rational solutions are in Casoratian form and the basic
column vector satisfies an extended condition equation set.

\end{abstract}


%\tableofcontents

%%%%%%%%%%%%%%%%%%%%%%%%%%%%%%%%%%%%%%%%%%%%%%%%%%%%%%%%%%%%%%%%%%%%%%%%%%%%%%%%%%%%%%%%%%%%%%%%%%%%%%%%%%%%%%%%%

\section{Introduction}

As one of interpretations of integrability of lattice equations,
multidimensional consistency \cite{Nijhoff-MDC} has become increasingly popular in the recent years.
With this property and two mild additional requirements on the equations: symmetry and
the so-called `tetrahedron property', Adler, Bobenko and Suris
classified the integrable models defined on an elementary quadrilateral \cite{ABS-CMP-2003}.
The corresponding result is named as ABS list, which consists of nine lattice equations:
${\rm Q4, Q3_{\del}, Q2, Q1_{\del}, A2,}$ ${\rm  A1_{\del}, H3_{\del}, H2, H1}$.
These equations are of form
\begin{subequations}\label{H/Q-list}
\begin{align}
\label{Q4} \mbox{Q4}:~~& p'(u\wt{u}+\wh{u}\wh{\wt{u}})-q'(u\wh{u}+\wt{u}\wh{\wt{u}}) \nn \\
~~&~~=\frac{ p'Q'-q'P'}{1-p'^2q'^2}\left((\wh{u}\wt{u}+u\wh{\wt{u}})-
p'q'(1+u\wt{u}\wh{u}\wh{\wt{u}})\right),\\
\label{Q3} \mbox{Q3}_{\del}:~~& p'(1-q'^2)(u\wh{u}+\wt{u}\wh{\wt{u}})-q'(1-p'^2)(u\wt{u}+\wh{u}\wh{\wt{u}}) \nn \\
~~&~~=(p'^2-q'^2)\left((\wh{u}\wt{u}+u\wh{\wt{u}})+\del^2\frac{(1-p'^2)(1-q'^2)}{4p'q'}\right), \\
\label{Q2} \mbox{Q2}:~~&p'(u-\wh{u})(\wt{u}-\wh{\wt{u}})-q'(u-\wt{u})(\wh{u}-\wh{\wt{u}})
+p'q'(p'-q')(u+\wt{u}+\wh{u}+\wh{\wt{u}})\nn\\
~~&~~=p'q'(p'-q')(p'^2-p'q'+q'^2),\\
\label{Q1} \mbox{Q1}_{\del}:~~& p'(u-\wh{u})(\wt{u}-\wh{\wt{u}})-q'
(u-\wt{u})(\wh{u}-\wh{\wt{u}})=\del^2p'q'(q'-p'),\\
\label{H3} \mbox{H3}_{\del}:~~& p'(u\wt{u}+\wh{u}\wh{\wt{u}})-q'(u\wh{u}+\wt{u}\wh{\wt{u}})=
\delta(q'^2-p'^2),\\
\label{H2} \mbox{H2}:~~&(u-\wh{\wt{u}})(\wt{u}-\wh{u})=(p'-q')(u+\wt{u}+
\wh{u}+\wh{\wt{u}})+p'^2-q'^2,\\
\label{H1} \mbox{H1}:~~&(u-\wh{\wt{u}})(\wt{u}-\wh{u})=p'-q',
\end{align}
\end{subequations}
where $P'^2=p'^4-\ga p'^2+1$ and $Q'^2=q'^4-\ga q'^2+1$.
We omit ${\rm A1_{\delta}}$ and ${\rm A2}$ from the above list
because of the equivalence between ${\rm A1_{\delta}}$ and ${\rm Q1}$ by
$u\to(-1)^{n+m}u$, as well as ${\rm A2}$ and ${\rm Q3_{\delta=0}}$ by $u\to
u^{(-1)^{n+m}}$. In equations \eqref{H/Q-list}, $\delta$ is a constant;
$u=u_{n,m}:=u(n,m)$ denotes the dependent
variable of the lattice points labeled by $(n,m)\in \mathbb{Z}^2$;
$p'$ and $q'$ are the continuous lattice parameters associated with the grid
size in the directions of the lattice given by the independent variables $n$ and $m$, respectively;
notations with elementary lattice shifts are denoted by
\[\wt{u}=u_{n+1,m},\quad \wh{u}=u_{n,m+1},\quad \wh{\wt u}=u_{n+1,m+1}.\]
%Additionally, we appoint that $\dt{}$ and $\dh{}$ , respectively, denote backward shifts with respect to $n$ and $m$, for
%example, $\dt u=u_{n-1,m},~~\dh u=u_{n,m-1}$.
Some of these equations have been known before, for example,
${\rm H1}$ is the lattice potential Korteweg de-Vries (lpKdV) equation \cite{NQC-PLA-1983},
${\rm H3_{\del=0}}$ is the lattice potential modified KdV (lpmKdV) equation \cite{NQC-PLA-1983},
${\rm Q1_{\del=0}}$ is the lattice Schwarzian KdV (lSKdV) equation \cite{NC-KdV} and
${\rm Q4}$ is known as the Adler's equation \cite{Adler-Q4}, which is the nonlinear superposition principle
for B\"{a}cklund transformation of the Krichever-Novikov equation.
There are many ways of degenerations among these lattice equations in the list \eqref{H/Q-list} \cite{ABS-CMP-2003,BT-At,Nijhoff-CM-2009}.

Various of approaches have been shown to be effective in deriving soliton
solutions for the ABS list \eqref{H/Q-list} as evidenced by series of papers.
Atkinson {\it et al.} constructed $N$-soliton solutions to ${\rm Q3_{\del}}$
in terms of the $\tau$-function of the Hirota-Miwa equation.
The corresponding solutions were expressed by the usual
Hirota's polynomial of exponentials \cite{AHN-Q3}. By developing Hirota's direct method, Hietarinta and Zhang derived
$N$-soliton solutions to H-series of equations and ${\rm Q1}$ \cite{HZ-ABS}.
This method is algorithmic and based on multidimensional consistency,
progressing in each case from background solution to 1-soliton solution to $N$-soliton solutions,
where many Casoratian shift formulae were established. Meanwhile, Nijhoff and his collaborators proposed Cauchy matrix approach \cite{Nijhoff-CM-2009}
to catch the $N$-soliton solutions for the ABS list except for the elliptic case of ${\rm Q4}$.
The authors of the present paper extended
Cauchy matrix approach to a generalized case \cite{ZZ-SAM-2013}, which can be
used to construct more kinds of exact solutions beyond soliton solutions for integrable systems (see also Ref. \cite{XZZ-KdV-SG}), such as,
multiple-pole solutions. By setting initial value problem, Inverse Scattering Transform was also established to
solve ${\rm H1}$ \cite{H1-IST} and the ABS list \cite{ABS-IST}. As the `master' and the most complicate equation in this list,
${\rm Q4}$ was solved by using the B\"{a}cklund transformation \cite{Atkinson-CMP}.

Different from soliton solution, rational solution is usually expressed by fraction of polynomials. Generally speaking,
such type of solutions can be derived from soliton solutions through a special limit
procedure (see Refs. \cite{AS-RS,ZDJ-Wronskian} as examples). Compared with the case in
continuous integrable system, it is more difficult to get the rational solutions of lattice equations.
In spite of this, until now much progress has been got.
Algebraic solutions and lump-like solutions for the Hirota-Miwa
equation were, respectively, given in Refs. \cite{Kaji} and \cite{Willox}.
With the help of bilinear method \cite{HZ-ABS}, rational solutions for ${\rm H3_{\delta}}$ and ${\rm Q1_{\delta}}$
as well as lattice Boussinesq equation were shown in recent
papers \cite{SZ-H3Q1,NZ-BSQ}. Besides, by imposing reduction conditions on rational
solutions for the Hirota-Miwa equation, rational solutions for lpKdV
equation and two semi-discrete lpKdV equations were obtained \cite{FZS}.
Recently, transformation approach was proposed to construct the rational solutions for the whole ABS list except for
Q3$_{\del}$ and ${\rm Q4}$ \cite{ZZ-ABS-RS}.

The present paper is devoted to investigating rational solutions for the ABS list
with a different methodology, where ${\rm Q4}$ is excluded.
The paper is organised as follows. In Sec.2, some necessary materials are displayed as preliminary, including the re-parameterized ABS lattices and Casoratian.
In Sec.3, we derive the rational solutions to lattice KdV-type equations involving lpKdV equation,
lpmKdV equation, NQC equation, together with two Miura transformations.
In Sec.4, rational solutions for ${\rm Q3_{\del}}$ are presented.
%We show that the form of soliton solutions to ${\rm Q3_{\del}}$ derived from the Cauchy matrix
%approach still can be used to describe the rational solutions.
Furthermore, degenerations will be
considered to derive the rational solutions for
``lower equations'' ${\rm Q2,Q1_{\del},H3_{\del},H2}$ and ${\rm H1}$.
Sec.5 is for conclusions. In addition, an appendix is given as a complement to the paper.

\section{Preliminary} \label{Sec-Preliminary}

Some new parameters are usually introduced such that the list \eqref{H/Q-list} can be handled easily.
For example, in Ref.\cite{Nijhoff-CM-2009} the ABS lattice equations \eqref{H/Q-list} except for ${\rm Q4}$ were re-parameterized
so that their solutions can be expressed through Cauchy matrices.
These re-parametrisations are given by \cite{Nijhoff-CM-2009}
\begin{equation}
\label{ABS-trans}
\begin{array}{lll}
\mbox{Q3}_{\del}: & ~~p'=\frac{P}{p^2-a^2}=\frac{p^2-b^2}{P}, &~~q'=\frac{Q}{q^2-a^2}=\frac{q^2-b^2}{Q}, \\
\mbox{Q2,Q1}_{\del}: & ~~p'=\frac{a^2}{p^2-a^2}, &~~q'=\frac{a^2}{q^2-a^2}, \\
\mbox{H3}_{\del}: & ~~ p'=\frac{P}{a^2-p^2}=\frac{1}{P}, &~~q'=\frac{Q}{a^2-q^2}=\frac{1}{Q}, \\
\mbox{H2,H1}: & ~~p'=-p^2, &~~q'=-q^2.
\end{array}
\end{equation}
And the re-parameterized lattice equations are
\begin{subequations}
\label{ABS-list-p}
\begin{align}
\mbox{Q3}_{\del}: &~~P(u\widehat{u}+\widetilde{u}\widehat{\widetilde{u}})
-Q(u\widetilde{u} +\widehat{u}\widehat{\widetilde{u}}) =(p^2-
q^2)\big((\widetilde{u}
\widehat{u}+u\widehat{\widetilde{u}})+\frac{\delta^2}{4PQ}\big),
\label{eq:Q3parm}\\
\mbox{Q2}: &~~(q^2-a^2)(u-\wh{u})(\wt{u}-\wh{\wt{u}})-(p^2-a^2)(u-\wt{u})(\wh{u}-\wh{\wt{u}})\nonumber\\
&~~~~ + (p^2-a^2)(q^2-a^2)(q^2-p^2)(u+\wt{u}+\wh{u}+\wh{\wt{u}})\nonumber\\
&~~ =(p^2-a^2)(q^2-a^2)(q^2-p^2)\big((p^2-a^2)^2+(q^2-a^2)^2-(p^2-a^2)(q^2-a^2)\big),
\label{Q2-p}\\
\mbox{Q1}_{\del}: &~~(q^2-a^2)(u-\wh{u})(\wt{u}-\wh{\wt{u}})-
(p^2-a^2)(u-\wt{u})(\wh{u}-\wh{\wt{u}}) =\frac{\delta^2a^4(p^2-q^2)}{(p^2-a^2)(q^2-a^2)},
\label{Q1-p}\\
\mbox{H3}_{\del}: &~~  P(a^2-q^2)(u\wt{u}+\wh{u}\wh{\wt{u}})-Q(a^2-p^2)
(u\wh{u}+\wt{u}\wh{\wt{u}})= \delta(p^2-q^2), \label{H3-p} \\
\mbox{H2}: &~~ (u-\wh{\wt{u}})(\wt{u}-\wh{u})+(p^2-q^2)
(u+\wt{u}+\wh{u}+\wh{\wt{u}})=p^4-q^4, \label{H2-p}\\
\mbox{H1}: &~~ (u-\wh{\wt{u}})(\wh{u}-\wt{u})=p^2-q^2, \label{H1-p}
\end{align}
\end{subequations}
where in \eqref{eq:Q3parm} $(p,P)=\mathfrak{p}$ and $(q,Q)=\mathfrak{q}$ are the points on the elliptic curve
\begin{equation}
 \{(x,X)|X^2 =(x^2-a^2)(x^2-b^2)\},
\label{elliptic-Q3}
\end{equation}
in \eqref{H3-p}
\begin{align}
P^2=a^2-p^2,~~ Q^2=a^2-q^2,
\label{eq:parcurves}
\end{align}
and in ${\rm Q3_{\del}}$ and ${\rm Q2}$ the dependent variable $u$ has been scaled by
\[u\to u(b^2-a^2),~~ u\to \frac{a^4 u}{(p^2-a^2)^2 (q^2-a^2)^2},\]
respectively. In the present paper, we focus on constructing rational solutions for the
ABS list \eqref{ABS-list-p}, where some Casoratian techniques developed in recent
literatures \cite{Kajiwara,HZ-ABS} will be adopted.

Casoratian can be viewed as the discrete version of Wronskian.
In general, for a given basic column vector
\begin{equation}\label{psi-vector}
\Ph(\al,\be,l)=(\phi_1(\al,\be,l),\phi_2(\al,\be,l),\cdots,\phi_N(\al,\be,l))^T,
\end{equation}
with $\{\phi_j(\al,\be,l)=\phi_j(n,m,\al,\be,l)\}$, the corresponding Casoratian can be written as
\begin{eqnarray}
\label{Caso}
f &=& |\Ph(\al,\be,0),\Ph(\al,\be,1),\ldots,\Ph(\al,\be,N-1)| \nn \\
~~~~~~~~~~~~~~~~~~~~~~~ &=& |0_{\al,\be},1_{\al,\be},\ldots,N-1_{\al,\be}|\nn \\
~~~~~~~~~~~~~~~~~~~~~~~ &=& |(0,1,\ldots,N-1)_{\al,\be}|.
\end{eqnarray}
Here and hereafter, we employ the  short-hand notations \cite{FN-KP-Wro}, such as
\begin{subequations}
\begin{eqnarray}
\label{Caso-exa}
|(0,1,\ldots,N-1)_{\al,\be}| &=& |(\wh{N-1})_{\al,\be}|, \\
|(0,1,\ldots,N-2,N)_{\al,\be}| &=& |(\wh{N-2},N)_{\al,\be}|.
%|(-1,0,2,\ldots,N-1)_{\alpha,\beta}| &=& |-1_{\alpha,\beta},0_{\alpha,\beta},(\overline{N-1})_{\alpha,\beta}|.
\end{eqnarray}
\end{subequations}
%To derive the rational solutions for the ABS list, we need to consider the following function
%\begin{eqnarray}
%\label{phi-c1}
%& \phi_j(\al,\be,l)=\rho^{(0)^{+}}_j(p+k_j)^n(q+k_j)^m(a+k_j)^{\al}(b+k_j)^{\be}(c+k_j)^{l} \nn \\
%&\qquad \qquad\qquad \qquad  + \rho^{(0)^{-}}_j(p-k_j)^n(q-k_j)^m(a-k_j)^{\al}(b-k_j)^{\be}(c-k_j)^{l},
%\end{eqnarray}
%where $\rho^{(0)^{\pm}}_j$ and $k_j$ are constants. This kind of function involving $(p,n)$ $(q,m)$, $(a,\al)$ and
%$(b,\be)$ simultaneous first appeared in Ref. \cite{AHN-Q3} (See also \cite{Nijhoff-CM-2009}).
%In \eqref{phi-c1} there exists symmetric relationship between $(p,n)$ $(q,m)$, $(a,\al)$ and $(b,\be)$, thus $a$ and $b$ can also be
%viewed as the lattice parameters.

In addition to the above notations, we need the following Laplace expansion
identity for Casoratian verification.
\begin{lemma} \cite{FN-KP-Wro}\label{Lemm-Babcd}
Suppose that $\mathbf{G}$ is a $N\times(N-2)$ matrix, and
\textbf{a},\textbf{b},\textbf{c},\textbf{d} are $N$th-order column
vectors, then
\begin{equation}
|\mathbf{G},\mathbf{a},\mathbf{b}||\mathbf{G},\mathbf{c},\mathbf{d}|
-|\mathbf{G},\mathbf{a},\mathbf{c}||\mathbf{G},\mathbf{b},\mathbf{d}|
+|\mathbf{G},\mathbf{a},\mathbf{d}||\mathbf{G},\mathbf{b},\mathbf{c}|=0.
\end{equation}
\end{lemma}

\section{Rational solutions for lattice KdV-type equations}

In this section, we first review the usual Casoratian solutions for the lattice KdV-type equations,
including lpKdV equation, lpmKdV equation, NQC equation and two Miura transformations. Then by introducing
an extended condition equation set, we derive rational solutions for the lattice KdV-type equations.
The results play key roles in the construction of the rational solutions for ${\rm Q3_{\del}}$.

\subsection{Casoratian solutions for lpKdV equation}

The lpKdV equation reads
\begin{eqnarray}
(p-q+\wh{w}-\wt{w})(p+q+w-\wh{\wt{w}})=p^2-q^2,
\label{lpKdV}
\end{eqnarray}
which is equivalent to equation \eqref{H1-p} by a change of dependent
variable $w=u+np+mq+u_0$ ($u_0$ is a constant) and reduces to the pKdV equation after a double
continuum limit. The lpKdV equation \eqref{lpKdV} admits bilinear form \cite{HZ-ABS}
\begin{subequations}
\begin{align}
&\mathcal{H}_1\equiv \wh g\wt f-\wt g \wh f+(p-q)(\wh f \wt f-f\wh{\wt f})=0,\label{lpKdV-Bili-1}\\
&\mathcal{H}_2\equiv g \wh{\wt f}-\wh{\wt g}f+(p+q)(f\wh{\wt f}-\wh f\wt f)=0 \label{lpKdV-Bili-2}
\end{align}
\label{lpKdV-Bili}
\end{subequations}
under dependent transformation
\begin{align}
w=\frac{g}{f}. \label{w}
\end{align}

Casoratian solutions to equation \eqref{lpKdV-Bili}
can be summarized by the following result \cite{ZH-H1-GE}.
\begin{prop}
\label{Pro-lpKdV-CS}
The bilinear equation \eqref{lpKdV-Bili} possesses Casoratians
\begin{equation}
f=|(\wh{N-1})_{0,0}|,~~g=|(\wh{N-2},N)_{0,0}|,
\label{fg-lpKdV}
\end{equation}
in which the column vector $\Ph(0,0,l)$ satisfies
\begin{subequations}
\begin{eqnarray}
p\Ph(\al,\be,l) &=& \wt{\Ph}(\al,\be,l)-\Ph(\al,\be,l+1), \label{eq:CES-1a} \\
q\Ph(\al,\be,l) &=& \wh{\Ph}(\al,\be,l)-\Ph(\al,\be,l+1), \label{eq:CES-1b} \\
p\Ps(\al,\be,l) &=& \wt{\Ps}(\al,\be,l)-\Ps(\al,\be,l+1), \label{eq:CES-1c} \\
q\wh{\Ps}(\al,\be,l) &=& \Ps(\al,\be,l)+\wh{\Ps}(\al,\be,l+1), \label{eq:CES-1d} \\
\Ph(\al,\be,l) &=& \bA_{[m]}\Ps(\al,\be,l), \label{eq:CES-1e}
\end{eqnarray}
\label{eq:CES-1}
\end{subequations}
where $\Ps(\al,\be,l)=(\psi_1(\al,\be,l),\psi_2(\al,\be,l),\ldots,\psi_N(\al,\be,l))^T$ is an auxiliary vector and $\bA_{[m]}$ is a
$N\times N$ matrix that only depends on $m$ but is independent of
$n,\al,\be$ and $l$.

\end{prop}

For the detailed proof, one can see Ref. \cite{ZH-H1-GE}.

\subsection{Casoratian solutions for lpmKdV equation}

The lpmKdV equation can be described as
\begin{eqnarray}
v_{a}((p-a)\wh{v}_{a}-(q-a)\wt{v}_{a})=
\wh{\wt{v}}_{a}((p+a)\wt{v}_{a}-(q+a)\wh{v}_{a}),
\label{lpmKdV}
\end{eqnarray}
where $a$ is a non-zero constant.
%or
%\begin{eqnarray}
%w_{a,b}((p-b)\wh{w}_{a,b}-(q-b)\wt{w}_{a,b})=
%\wh{\wt{w}}_{a,b}((p+b)\wt{w}_{a,b}-(q+b)\wh{w}_{a,b}),
%\label{lpmKdV-2}
%\end{eqnarray}
Under transformation
\begin{align}
v_{a}=\frac{h}{f},
\label{va}
\end{align}
equation \eqref{lpmKdV} is bilinearized into
\begin{subequations}
\begin{align}
&\mathcal{H}_{11}\equiv(p-a)\wh{\wt{f}}h+(q+a)f\wh{\wt{h}}-(p+q)\wh{f}\wt{h}=0,\label{lpmKdV-Bili-1a}\\
&\mathcal{H}_{12}\equiv(p+a)f\wh{\wt{h}}+(q-a)\wh{\wt{f}}h-(p+q)\wt{f}\wh{h}=0,\label{lpmKdV-Bili-1b}
\end{align}
\label{lpmKdV-Bili-1}
\end{subequations}
or
\begin{subequations}
\begin{align}
&\mathcal{H}_{21}\equiv(p-a)\wt f\wh h-(q-a)\wh f \wt h-(p-q)f\wh{\wt{h}}=0,\label{lpmKdV-Bili-2a}\\
&\mathcal{H}_{22}\equiv(p+a)\wh f\wt h-(q+a)\wt f \wh
h-(p-q)\wh{\wt{f}}h=0.\label{lpmKdV-Bili-2b}
\end{align}
\label{lpmKdV-Bili-2}
\end{subequations}

The Casoratian solutions to these two bilinear equations are shown as follows.
\begin{prop}
\label{Pro-lpmKdV-CS}
The bilinear equations \eqref{lpmKdV-Bili-1} and \eqref{lpmKdV-Bili-2} possess Casoratians
\begin{equation}
f=|(\wh{N-1})_{0,0}|,~~h=a^N|(\wh{N-1})_{-1,0}|,
\label{fh-lpmKdV}
\end{equation}
in which the column vectors $\Ph(0,0,l)$ and $\Ph(-1,0,l)$ satisfy system \eqref{eq:CES-1}
together with
\begin{subequations}
\begin{eqnarray}
a\Ph(\al,\be,l) &=& \cd{\Ph}(\al,\be,l)-\Ph(\al,\be,l+1), \label{eq:CES-2a} \\
a\Ps(\al,\be,l) &=& \cd{\Ps}(\al,\be,l)-\Ps(\al,\be,l+1), \label{eq:CES-2b} \\
p\wt{\Vph}(\al,\be,l) &=& \Vph(\al,\be,l)+\wt{\Vph}(\al,\be,l+1), \label{eq:CES-2c} \\
q\Vph(\al,\be,l) &=& \wh{\Vph}(\al,\be,l)-\Vph(\al,\be,l+1), \label{eq:CES-2d} \\
a\Vph(\al,\be,l) &=& \cd{\Vph}(\al,\be,l)-\Vph(\al,\be,l+1), \label{eq:CES-2e} \\
\Ph(\al,\be,l) &=& \bB_{[n]}\Vph(\al,\be,l), \label{eq:CES-2f}
\end{eqnarray}
\label{eq:CES-2}
\end{subequations}
where $\Vph(\al,\be,l)=(\varphi_1(\al,\be,l),\varphi_2(\al,\be,l),\ldots,\varphi_N(\al,\be,l))^T$ is an auxiliary vector and~
$\cd{\phantom{a}}$ denotes shift with respect to $\al$, i.e., $\cd{\Ph}(\al,\be,l)=\Ph(\al+1,\be,l)$;
$\bB_{[n]}$ is a $N\times N$ matrix that only depends on $n$ but is independent of
$m,\al,\be$ and $l$.
\end{prop}

\begin{proof}

Because any three bilinear equations in \eqref{lpmKdV-Bili-1}
and \eqref{lpmKdV-Bili-2} can lead to the remainder,
here we just study equations $\mathcal{H}_{11}$, $\mathcal{H}_{12}$ and
$\mathcal{H}_{21}$.

We use shifted $\mathcal{H}_{11}$ in the following form
\begin{eqnarray}
\dt{\mathcal{H}}_{11}\equiv(p-a)\wh{f}\dt h+(q+a)\dt f\wh{h}-(p+q)\wh{\dt{f}}h=0. \label{lpmKdV-Bili-1a-uth}
\end{eqnarray}
By virtue of equation \eqref{eq:CES-2a}, it is easy to know that $h$ in \eqref{fh-lpmKdV}
can be rewritten as
\begin{eqnarray}
\label{eq:lpmKP-fh-solu-re}
h=a^N|0_{-1,0},(\wh{N-2})_{0,0}|.\label{h-lpmKdV}
\end{eqnarray}
In \eqref{lpmKdV-Bili-1a-uth}, for $\dt f$, $\dt h$, $\wh{f}$, $\wh{h}$ and $\wh{\dt{f}}$,
we make use of \eqref{eq:lpKdV-iden-1a}, \eqref{eq:lpmKdV-iden-1c},
\eqref{eq:lpKdV-iden-2b}, \eqref{eq:lpKdV-iden-2d} and \eqref{eq:lpKdV-iden-2e}, respectively, and get
\begin{eqnarray}
&& (pq)^{N-2}[(p-a)\wh{f}\dt h+(q+a)\dt f\wh{h}-(p+q)\wh{\dt{f}}h] \nn \\
&&~~~=a^N\frac{|\wh{\bA}_{[m]}|}{|\bA_{[m]}|}
\big(|0_{-1,0},(\wh{N-3})_{0,0},\dt{\Ph}(0,0,N-2)|
|(\wh{N-2})_{0,0},\bA_{[m]}\wh{\bA}_{[m]}^{-1}\wh{\Ph}(0,0,N-2)|\nn \\
&& ~~~~~-|(\wh{N-2})_{0,0},\dt{\Ph}(0,0,N-2)||0_{-1,0},(\wh{N-3})_{0,0},\bA_{[m]}\wh{\bA}_{[m]}^{-1}\wh{\Ph}(0,0,N-2)| \nn \\
&& ~~~~~-|(\wh{N-3})_{0,0},\dt{\Ph}(0,0,N-2),\bA_{[m]}\wh{\bA}_{[m]}^{-1}\wh{\Ph}(0,0,N-2)|
|0_{-1,0},(\wh{N-2})_{0,0}\big),
\end{eqnarray}
which vanishes in the light of Lemma \ref{Lemm-Babcd}, where $\mathbf{G}=(\wh{N-3})_{0,0}$,
$(\mathbf{a},\mathbf{b},\mathbf{c},\mathbf{d})=(0_{-1,0},\dt{\Ph}(0,0,N-2),N-2_{0,0},\bA_{[m]}\wh{\bA}_{[m]}^{-1}\wh{\Ph}(0,0,N-2))$.

Similarly, for proving $\mathcal{H}_{12}$, we consider its
shifted form
\begin{eqnarray}
\dh{\mathcal{H}}_{12}\equiv(p+a)\dh f\wt{h}+(q-a)\wt{f}\dh h-(p+q)\dh{\wt{f}}h=0.
\label{lpmKdV-Bili-1b-uth}
\end{eqnarray}
In \eqref{lpmKdV-Bili-1b-uth}, for $\dh f$, $\dh h$, $\wt{f}$, $\wt{h}$ and $\wt{\dh{f}}$,
we use \eqref{eq:lpKdV-iden-1b}, \eqref{eq:lpmKdV-iden-1d},
\eqref{eq:lpKdV-iden-2a}, \eqref{eq:lpKdV-iden-2c} and \eqref{eq:lpKdV-iden-2f}, respectively. Then we have
\begin{align}
 \dh{\mathcal{H}}_{12} \equiv
 &(p+a)\dh f\wt{h}+(q-a)\wt{f}\dh h-(p+q)\dh{\wt{f}}h \nn \\
= & a^N(pq)^{-N+2}\frac{|\wt{\bB}_{[n]}|}{|\bB_{[n]}|}
\big(-|0_{-1,0},(\wh{N-3})_{0,0},\bB_{[n]}\wt{\bB}_{[n]}^{-1}\wt{\Ph}(0,0,N-2)||(\wh{N-2})_{0,0},\dh{\Ph}(0,0,N-2)|\nn \\
&+|(\wh{N-2})_{0,0},\bB_{[n]}\wt{\bB}_{[n]}^{-1}\wt{\Ph}(0,0,N-2)||0_{-1,0},(\wh{N-3})_{0,0},\dh{\Ph}(0,0,N-2)| \nn
\\
&-|(\wh{N-3})_{0,0},\dh{\Ph}(0,0,N-2),\bB_{[n]}\wt{\bB}_{[n]}^{-1}\wt{\Ph}(0,0,N-2)|
|0_{-1,0},(\wh{N-2})_{0,0}|\big)\nn \\
= & 0,
\end{align}
where we have utilized Lemma \ref{Lemm-Babcd}, in which $\mathbf{G}=(\wh{N-3})_{0,0}$,
$(\mathbf{a},\mathbf{b},\mathbf{c},\mathbf{d})=(0_{-1,0},\bB_{[n]}\wt{\bB}_{[n]}^{-1}\wt{\Ph}(0,0,N-2),N-2_{0,0},\dh{\Ph}(0,0,N-2))$.

Next, we adopt the down-tilde-hat version of $\mathcal{H}_{21}$, i.e.,
\begin{eqnarray}
\dh{\dt{\mathcal{H}}}_{21}\equiv(p-a)\dh f\dt{h}-(q-a)\dt{f}\dh h-(p-q)\dh{\dt{f}}h=0.
\label{lpmKdV-Bili-2a-uth}
\end{eqnarray}
For \eqref{lpmKdV-Bili-2a-uth}, $\dt f$, $\dh f$, $\dt{h}$, $\dh{h}$, $\dh{\dt{f}}$
are provided by \eqref{eq:lpKdV-iden-1a}, \eqref{eq:lpKdV-iden-1b}, \eqref{eq:lpmKdV-iden-1c},
\eqref{eq:lpmKdV-iden-1d} and \eqref{eq:lpKdV-iden-1g}, respectively. Now we obtain
\begin{eqnarray}
&& \dh{\mathcal{H}}_{21}\equiv(p-a)\dh f\dt{h}-(q-a)\dt{f}\dh h-(p-q)\dh{\dt{f}}h \nn \\
&&~~~~~=a^N(pq)^{-N+2}\big(-|(\wh{N-2})_{0,0},\dh{\Ph}(0,0,N-2)|
|0_{-1,0},(\wh{N-3})_{0,0},\dt{\Ph}(0,0,N-2)|\nn \\
&& ~~~~~~~+|(\wh{N-2})_{0,0},\dt{\Ph}(0,0,N-2)|
|0_{-1,0},(\wh{N-3})_{0,0},\dh{\Ph}(0,0,N-2)| \nn \\
&& ~~~~~~~-|(\wh{N-3})_{0,0},\dh{\Ph}(0,0,N-2),\dt{\Ph}(0,0,N-2)|
|0_{-1,0},(\wh{N-2})_{0,0}|\big)=0
\end{eqnarray}
by using Lemma \ref{Lemm-Babcd}, in which $\mathbf{G}=(\wh{N-3})_{0,0}$,
$(\mathbf{a},\mathbf{b},\mathbf{c},\mathbf{d})=(N-2_{0,0},\dh{\Ph}(0,0,N-2),0_{-1,0},\dt{\Ph}(0,0,N-2))$.
Thus we complete the proof.

%\vskip 5pt
%\noindent{\bf Remark:} If we consider equation \eqref

\end{proof}

The lpKdV equation \eqref{lpKdV} is related to the lpmKdV equation \eqref{lpmKdV} by
\begin{subequations}
\begin{align}
p-q+\wh{w}-\wt{w}= & \frac{1}{\wh{\wt{v}}_{a}}((p-a)\wh{v}_{a}-(q-a)\wt{v}_{a}) \label{lpKdV-lpmKdV-MT-1a}\\
= & \frac{1}{v_{a}}((p+a)\wt{v}_{a,b}-(q+a)\wh{v}_{a}), \label{lpKdV-lpmKdV-MT-1b} \\
p+q+w-\wh{\wt{w}}= & \frac{1}{\wt{v}_{a}}((p-a)v_{a}+(q+a)\wh{\wt{v}}_{a})\label{lpKdV-lpmKdV-MT-2a}\\
= & \frac{1}{\wh{v}_{a}}((p+a)\wh{\wt{v}}_{a}+(q-a)v_{a}),\label{lpKdV-lpmKdV-MT-2b}
\end{align}
\label{lpKdV-lpmKdV-MT}
\end{subequations}
which serve as the Miura transformation \cite{NC-KdV}.
Substituting dependent transformations \eqref{w} and \eqref{va} into \eqref{lpKdV-lpmKdV-MT},
one can easily find that \eqref{lpKdV-Bili}, \eqref{lpmKdV-Bili-1} and \eqref{lpmKdV-Bili-2} compose
the bilinear forms for the system \eqref{lpKdV-lpmKdV-MT}. Therefore, system \eqref{lpKdV-lpmKdV-MT}
have Casoratian solutions \eqref{w} and \eqref{va},
in which $f,g$ and $h$ are given by \eqref{fg-lpKdV} and \eqref{fh-lpmKdV},
where the basic column vector $\Ph(\al,\be,l)$ satisfies systems \eqref{eq:CES-1}
and \eqref{eq:CES-2}.

\subsection{Casoratian solutions for NQC equation}

The NQC equation was firstly introduced in Ref. \cite{NQC-PLA-1983} by direct linearization method.
This equation has the form
\begin{eqnarray}
\frac{1+(p-a)S(a,b)-(p+b)\wt{S}(a,b)}
{1+(q-a)S(a,b)-(q+b)\wh{S}(a,b)}
=\frac{1-(q+a)\wh{\wt{S}}(a,b)+(q-b)\wt{S}(a,b)}
{1-(p+a)\wh{\wt{S}}(a,b)+(p-b)\wh{S}(a,b)},
\label{eq:NQC}
\end{eqnarray}
where $a$ and $b$ are non-zero constants. Through different parameter choices, \eqref{eq:NQC} can
yield lpKdV equation \eqref{lpKdV} and lpmKdV equation \eqref{lpmKdV}.
In Ref. \cite{NQC-PLA-1983}, it was revealed that there is a Miura transformation
between lpmKdV equation \eqref{lpmKdV} and NQC equation \eqref{eq:NQC}, which is given by
\begin{subequations}
\begin{align}
& 1+(p-a)S(a,b)-(p+b)\wt{S}(a,b)=\wt{v}_av_b, \label{Sab-vawb-1} \\
& 1+(q-a)S(a,b)-(q+b)\wh{S}(a,b)=\wh{v}_av_b, \label{Sab-vawb-2} \\
& 1+(p-b)S(a,b)-(p+a)\wt{S}(a,b)=v_a\wt{v}_b, \label{Sab-vawb-3} \\
& 1+(q-b)S(a,b)-(q+a)\wt{S}(a,b)=v_a\wh{v}_b, \label{Sab-vawb-4}
\end{align}
\label{Sab-vawb}
\end{subequations}
where $v_a$ satisfies the lpmKdV equation \eqref{lpmKdV}.
Equation \eqref{eq:NQC} can be derived from \eqref{Sab-vawb} in the
light of equality $\frac{\wt{v}_av_b}{\wh{v}_av_b}
=\frac{(v_a\wh{v}_b)\wt{\phantom{a}}}{(v_a\wt{v}_b)\wh{\phantom{a}}}$.
Now rather than discussing equation \eqref{eq:NQC}, we turn to
consider system \eqref{Sab-vawb}. Through
transformations
\begin{align}
v_a=\frac{h}{f},~~v_b=\frac{s}{f},~~
S(a,b)=\frac{\tta}{f},
\label{NQC-tran}
\end{align}
system \eqref{Sab-vawb} can be bilinearized into
\begin{subequations}
\begin{align}
& \mathcal{H}_{31}\equiv f\wt{f}+(p-a)\tta \wt{f}-(p+b)\wt{\tta}f-\wt{h}s=0, \label{Sab-vawb-bi-a} \\
& \mathcal{H}_{32}\equiv f\wh{f}+(q-a)\tta \wh{f}-(q+b)\wh{\tta}f-\wh{h}s=0, \label{Sab-vawb-bi-b} \\
& \mathcal{H}_{33}\equiv f\wt{f}+(p-b)\tta \wt{f}-(p+a)\wt{\tta}f-h\wt{s}=0, \label{Sab-vawb-bi-c} \\
& \mathcal{H}_{34}\equiv f\wh{f}+(q-b)\tta \wh{f}-(q+a)\wh{\tta}f-h\wh{s}=0. \label{Sab-vawb-bi-d}
\end{align}
\label{Sab-vawb-bi}
\end{subequations}

For the Casoratian solutions to \eqref{Sab-vawb-bi}, one has
\begin{prop} \label{Pro-NQC-CS}
Bilinear system \eqref{Sab-vawb-bi} admits the solutions
\begin{eqnarray}
\label{eq:NQC-solu}
&& f=|(\wh{N-1})_{0,0}|,~~h=a^N|(\wh{N-1})_{-1,0}|,~~s=b^N|(\wh{N-1})_{0,-1}|, \nn \\
&& \tta=-\frac{1}{a+b}((ab)^N|(\wh{N-1})_{-1,-1}|-|(\wh{N-1})_{0,0}|), \label{eq:NQC-solu-b}
\end{eqnarray}
in which the basic column vector $\Ph(\al,\be,l)$ satisfies \eqref{eq:CES-1a}, \eqref{eq:CES-1b},
\eqref{eq:CES-2a} and
\begin{subequations}
\begin{eqnarray}
b\Ph(\al,\be,l) &=& \d{\Ph}(\al,\be,l)-\Ph(\al,\be,l+1), \label{eq:CES-3a} \\
p\Ch(\al,\be,l) &=& \wt{\Ch}(\al,\be,l)-\Ch(\al,\be,l+1), \label{eq:CES-3b} \\
q\Ch(\al,\be,l) &=& \wh{\Ch}(\al,\be,l)-\Ch(\al,\be,l+1), \label{eq:CES-3c} \\
a\Ch(\al,\be,l) &=& \cd{\Ch}(\al,\be,l)-\Ch(\al,\be,l+1), \label{eq:CES-3d} \\
b\d{\Ch}(\al,\be,l) &=& \Ch(\al,\be,l)+\d{\Ch}(\al,\be,l+1), \label{eq:CES-3e} \\
\Ph(\al,\be,l) &=& \bC_{[\be]}\Ch(\al,\be,l), \label{eq:CES-3f} \\
p\Vap(\al,\be,l) &=& \wt{\Vap}(\al,\be,l)-\Vap(\al,\be,l+1), \label{eq:CES-4a} \\
q\Vap(\al,\be,l) &=& \wh{\Vap}(\al,\be,l)-\Vap(\al,\be,l+1), \label{eq:CES-4b} \\
a\cd{\Vap}(\al,\be,l) &=& \Vap(\al,\be,l)+\cd{\Vap}(\al,\be,l+1), \label{eq:CES-4c} \\
b\Vap(\al,\be,l) &=& \d{\Vap}(\al,\be,l)-\Vap(\al,\be,l+1), \label{eq:CES-4d} \\
\Ph(\al,\be,l) &=& \bD_{[\al]}\Vap(\al,\be,l), \label{eq:CES-4e}
\end{eqnarray}
\label{eq:CES-3-4}
\end{subequations}
where $\Ch(\al,\be,l)=(\chi_1(\al,\be,l),\chi_2(\al,\be,l),\ldots,\chi_N(\al,\be,l))^T$ and
$\Vap(\al,\be,l)=(\varpi_1(\al,\be,l),\varpi_2(\al,\be,l),$ $\ldots,\varpi_N(\al,\be,l))^T$ are two auxiliary vectors and
$\d{\phantom{a}}$ denotes shift with respect to $\be$, i.e., $\d{\Ph}(\al,\be,l)=\Ph(\al,\be+1,l)$;
$\bC_{[\be]}$ is a $N\times N$ matrix that only depends on $\be$ but is independent of
$n,m,\al$ and $l$; $\bD_{[\al]}$ is a $N\times N$ matrix only depends on $\al$ but is independent of
$n,m,\be$ and $l$. In \eqref{eq:NQC-solu-b} we add factor $(ab)^N$ in $\tta$
to formally avoid the constraint $b \neq -a$. One can see the explicit expressions for the two simplest
rational solutions listed in Sec.\ref{RS-lKdV-type}.
\end{prop}

\begin{proof}

Noting that the structure of $\theta$, we rewrite \eqref{Sab-vawb-bi} to
\begin{subequations}
\label{eq:Sab-bi-m}
\begin{eqnarray}
&& \mathcal{H}_{41}\equiv(p-a)\wt f \vta-(p+b)\wt{\vta} f+(a+b)\wt{h}s=0, \label{eq:Sab-bi-m-a} \\
&& \mathcal{H}_{42}\equiv(q-a)\wh f \vta-(q+b)\wh{\vta} f+(a+b)\wh{h}s=0, \label{eq:Sab-bi-m-b} \\
&& \mathcal{H}_{43}\equiv(p-b)\wt f \vta-(p+a)\wt{\vta} f+(a+b)h\wt{s}=0, \label{eq:Sab-bi-m-c} \\
&& \mathcal{H}_{44}\equiv(q-b)\wh f \vta-(q+a)\wh{\vta} f+(a+b)h\wh{s}=0, \label{eq:Sab-bi-m-d}
\end{eqnarray}
\end{subequations}
where
\begin{eqnarray*}
\vta=(ab)^N|(\wh{N-1})_{-1,-1}|.
\end{eqnarray*}
For proving \eqref{eq:Sab-bi-m-a}, we make use of the replacements
$f=\frac{1}{b^N}\d{s}$ and $h=\frac{a^N}{b^N}\d{\dc{s}}$ and consider equation
\begin{eqnarray}
\dt{\mathcal{H}}_{41}\equiv(p-a) \d{s} \dt{\vta}-(p+b)\vta \d{\dt{s}}+a^N(a+b)\d{\dc{s}}\dt{s}=0, \label{eq:Sab-bi-m-a1}
\end{eqnarray}
where $\vta=(ab)^N|0_{-1,-1},(\wh{N-2})_{0,-1}|$.
In \eqref{eq:Sab-bi-m-a1}, for $\d{s}$, $\dt{\vta}$, $\d{\dt{s}}$, $\d{\dc{s}}$ and $\dt{s}$,
we use \eqref{eq:NQC-iden-2a}, \eqref{eq:NQC-iden-1c},
\eqref{eq:NQC-iden-2b}, \eqref{eq:NQC-iden-2d} and \eqref{eq:NQC-iden-1a}, respectively. Then we have
\begin{eqnarray}
&& \dt{\mathcal{H}}_{41} \equiv (p-a)\d{s} \dt{\vta}-(p+b)\vta \d{\dt{s}}+a^N(a+b)\d{\dc{s}}\dt{s} \nn \\
&&~~~~~=(ab^2)^N(pb)^{-N+2}\frac{|\d{\bC}_{[\be]}|}{|\bC_{[\be]}|}
\big(|0_{-1,-1},(\wh{N-3})_{0,-1},\dt{\Ph}(0,-1,N-2)|\nn \\
&& ~~~~~~~|(\wh{N-2})_{0,-1},\bC_{[\be]}\d{\bC}_{[\be]}^{-1}\d{\Ph}(0,-1,N-2)|-
|(\wh{N-3})_{0,-1},\dt{\Ph}(0,-1,N-2),\nn \\
&& ~~~~~~~\bC_{[\be]}\d{\bC}_{[\be]}^{-1}\d{\Ph}(0,-1,N-2)|
|0_{-1,-1},(\wh{N-2})_{0,-1}| \nn \\
&& ~~~~~~~-|(\wh{N-2})_{0,-1},\dt{\Ph}(0,-1,N-2)|
|0_{-1,-1},(\wh{N-2})_{0,-1},\bC_{[\be]}\d{\bC}_{[\be]}^{-1}\d{\Ph}(0,-1,N-2)|\big)\nn \\
&& ~~~~~=0,
\end{eqnarray}
where Lemma \ref{Lemm-Babcd} was considered, in which $\mathbf{G}=(\wh{N-3})_{0,-1}$,
$(\mathbf{a},\mathbf{b},\mathbf{c},\mathbf{d})=(0_{-1,-1},\dt{\Ph}(0,-1,N-2),N-2_{0,-1},
\bC_{[\be]}\d{\bC}_{[\be]}^{-1}\d{\Ph}(0,-1,N-2))$. Similarly, one can prove
bilinear equations \eqref{eq:Sab-bi-m-b}-\eqref{eq:Sab-bi-m-d}. Thus we finish the verification.

\end{proof}

Let us conclude Propositions \ref{Pro-lpKdV-CS}-\ref{Pro-NQC-CS} as follows.
\begin{prop}
\label{Pro-lpKdV-type-CS-1}
The Casoratians
\begin{eqnarray}
&& f=|(\wh{N-1})_{0,0}|,~~g=|(\wh{N-2},N)_{0,0}|,~~h=a^N|(\wh{N-1})_{-1,0}|,~~s=b^N|(\wh{N-1})_{0,-1}|, \nn \\
&& \tta=-\frac{1}{a+b}((ab)^N|(\wh{N-1})_{-1,-1}|-|(\wh{N-1})_{0,0}|)
\label{fghstta}
\end{eqnarray}
solve the bilinear equations \eqref{lpKdV-Bili}, \eqref{lpmKdV-Bili-1},
\eqref{lpmKdV-Bili-2} and \eqref{Sab-vawb-bi}, where the basic column vector $\Ph(\al,\be,l)$ satisfies the
condition equation set \eqref{eq:CES-1}, \eqref{eq:CES-2}, \eqref{eq:CES-3-4} together with
\begin{subequations}
\label{eq:CES-2be}
\begin{eqnarray}
b\Ps(\al,\be,l) &=& \d{\Ps}(\al,\be,l)-\Ps(\al,\be,l+1), \label{eq:CES-2b-b} \\
b\Vph(\al,\be,l) &=& \d{\Vph}(\al,\be,l)-\Vph(\al,\be,l+1). \label{eq:CES-2e-b}
\end{eqnarray}
\end{subequations}

\end{prop}
The additional equations \eqref{eq:CES-2be} implies that $v_b$ given by \eqref{NQC-tran}
also satisfies the lpmKdV equation \eqref{lpmKdV} with $a\rightarrow b$.
According to the different forms of $\bA_{[m]}$,  $\bB_{[n]}$,  $\bC_{[\be]}$ and $\bD_{[\al]}$,
at least two types of solutions can be obtained (cf. Ref. \cite{ZH-H1-GE}). For example, when $\bA_{[m]}$, $\bB_{[n]}$, $\bC_{[\be]}$ and $\bD_{[\al]}$ are, respectively,
diagonal matrices defined as
\begin{subequations}
\label{ABCD}
\begin{eqnarray}
&& \bA_{[m]}=\mbox{Diag}((q^2-k_j^2)^m)_{N\times N},~~\bB_{[n]}=\mbox{Diag}((p^2-k_j^2)^n)_{N\times N}, \\
&& \bC_{[\be]}=\mbox{Diag}((b^2-k_j^2)^{\be})_{N\times N}, ~~ \bD_{[\al]}=\mbox{Diag}((a^2-k_j^2)^{\al})_{N\times N}.
\end{eqnarray}
\end{subequations}
Then \eqref{fghstta} together with
\begin{subequations}
\label{phi-c}
\begin{eqnarray}
&& \phi_j(\al,\be,l)=\rho^{(0)^{+}}_j(p+k_j)^n(q+k_j)^m(a+k_j)^{\al}(b+k_j)^{\be}k_j^{l} \nn \\
&& \qquad \qquad\qquad +\rho^{(0)^{-}}_j(p-k_j)^n(q-k_j)^m(a-k_j)^{\al}(b-k_j)^{\be}(-k_j)^{l}, \label{phi-ca} \\
&& \psi_j(\al,\be,l)=(q^2-k_j^2)^{-m}\phi_j(\al,\be,l), ~~ \varphi_j(\al,\be,l)=(p^2-k_j^2)^{-n}\phi_j(\al,\be,l), \\
&& \chi_j(\al,\be,l)=(b^2-k_j^2)^{-\be}\phi_j(\al,\be,l), ~~ \varpi_j(\al,\be,l)=(a^2-k_j^2)^{-\al}\phi_j(\al,\be,l)
\end{eqnarray}
\end{subequations}
for $j=1,2,\ldots,N$ provides the usual multi-soliton solutions for the lattice KdV-type equations, where $\{\rho^{(0)^{\pm}}_j\}$ are constants.

\subsection{Rational solutions for lattice KdV-type equations}
\label{RS-lKdV-type}

For $\rm{H3}_{\del}$ and $\rm{Q1}_{\del}$ in the ABS list, the existence of $\del \neq 0$ plays a crucial role
in the procedure of obtaining rational solutions from their soliton solutions \cite{SZ-H3Q1}. While for $\rm{H1}$ and $\rm{H2}$,
 not involving $\del$, it does not work to derive their rational solution from the soliton solutions \eqref{fghstta} with
\eqref{phi-c} by taking limit. To avoid this shortcoming, we consider the following system
\begin{subequations}
\begin{eqnarray}
(p-c)\Ph(\al,\be,l) &=& \wt{\Ph}(\al,\be,l)-\Ph(\al,\be,l+1), \label{eq:ex-CES-1a} \\
(q-c)\Ph(\al,\be,l) &=& \wh{\Ph}(\al,\be,l)-\Ph(\al,\be,l+1), \label{eq:ex-CES-1b}  \\
(a-c)\Ph(\al,\be,l) &=& \cd{\Ph}(\al,\be,l)-\Ph(\al,\be,l+1), \label{eq:ex-CES-1c} \\
(b-c)\Ph(\al,\be,l) &=& \d{\Ph}(\al,\be,l)-\Ph(\al,\be,l+1), \label{eq:ex-CES-1d} \\
(p-c)\Ps(\al,\be,l) &=& \wt{\Ps}(\al,\be,l)-\Ps(\al,\be,l+1), \label{eq:ex-CES-2a} \\
(q+c)\wh{\Ps}(\al,\be,l) &=& \Ps(\al,\be,l)+\wh{\Ps}(\al,\be,l+1), \label{eq:ex-CES-2b} \\
(a-c)\Ps(\al,\be,l) &=& \cd{\Ps}(\al,\be,l)-\Ps(\al,\be,l+1), \label{eq:ex-CES-2c} \\
(b-c)\Ps(\al,\be,l) &=& \d{\Ps}(\al,\be,l)-\Ps(\al,\be,l+1), \label{eq:ex-CES-2d} \\
\Ph(\al,\be,l) &=& \bA_{[m]}\Ps(\al,\be,l), \label{eq:ex-CES-PhPs} \\
(p+c)\wt{\Vph}(\al,\be,l) &=& \Vph(\al,\be,l)+\wt{\Vph}(\al,\be,l+1), \label{eq:ex-CES-3a} \\
(q-c)\Vph(\al,\be,l) &=& \wh{\Vph}(\al,\be,l)-\Vph(\al,\be,l+1), \label{eq:ex-CES-3b} \\
(a-c)\Vph(\al,\be,l) &=& \cd{\Vph}(\al,\be,l)-\Vph(\al,\be,l+1), \label{eq:ex-CES-3c} \\
(b-c)\Vph(\al,\be,l) &=& \d{\Vph}(\al,\be,l)-\Vph(\al,\be,l+1), \label{eq:ex-CES-3d} \\
\Ph(\al,\be,l) &=& \bB_{[n]}\Vph(\al,\be,l), \label{eq:ex-CES-PhVph} \\
(p-c)\Ch(\al,\be,l) &=& \wt{\Ch}(\al,\be,l)-\Ch(\al,\be,l+1), \label{eq:ex-CES-4a} \\
(q-c)\Ch(\al,\be,l) &=& \wh{\Ch}(\al,\be,l)-\Ch(\al,\be,l+1), \label{eq:ex-CES-4b} \\
(a-c)\Ch(\al,\be,l) &=& \cd{\Ch}(\al,\be,l)-\Ch(\al,\be,l+1), \label{eq:ex-CES-4c} \\
(b+c)\d{\Ch}(\al,\be,l) &=& \Ch(\al,\be,l)+\d{\Ch}(\al,\be,l+1), \label{eq:ex-CES-4d} \\
\Ph(\al,\be,l) &=& \bC_{[\be]}\Ch(\al,\be,l), \label{eq:ex-CES-PhCh} \\
(p-c)\Vap(\al,\be,l) &=& \wt{\Vap}(\al,\be,l)-\Vap(\al,\be,l+1), \label{eq:ex-CES-5a} \\
(q-c)\Vap(\al,\be,l) &=& \wh{\Vap}(\al,\be,l)-\Vap(\al,\be,l+1), \label{eq:ex-CES-5b} \\
(a+c)\cd{\Vap}(\al,\be,l) &=& \Vap(\al,\be,l)+\cd{\Vap}(\al,\be,l+1), \label{eq:ex-CES-5c} \\
(b-c)\Vap(\al,\be,l) &=& \d{\Vap}(\al,\be,l)-\Vap(\al,\be,l+1), \label{eq:ex-CES-5d} \\
\Ph(\al,\be,l) &=& \bD_{[\al]}\Vap(\al,\be,l), \label{eq:ex-CES-PhVap}
\end{eqnarray}
\label{eq:ex-CES}
\end{subequations}
where $\bA_{[m]}$, $\bB_{[n]}$, $\bC_{[\be]}$ and $\bD_{[\al]}$ are defined
as in Propositions \ref{Pro-lpKdV-CS}-\ref{Pro-NQC-CS}.
System \eqref{eq:ex-CES} is referred to as the extended condition equation set. For \eqref{ABCD}, solutions to \eqref{eq:ex-CES} can be described as
\begin{subequations}
\label{phi-ce}
\begin{eqnarray}
&& \phi_j(\al,\be,l)=\rho^{(0)^{+}}_j(p+k_j)^n(q+k_j)^m(a+k_j)^{\al}(b+k_j)^{\be}(c+k_j)^{l} \nn \\
&& \qquad \qquad\qquad +\rho^{(0)^{-}}_j(p-k_j)^n(q-k_j)^m(a-k_j)^{\al}(b-k_j)^{\be}(c-k_j)^{l}, \\
&& \psi_j(\al,\be,l)=(q^2-k_j^2)^{-m}\phi_j(\al,\be,l), ~~ \varphi_j(\al,\be,l)=(p^2-k_j^2)^{-n}\phi_j(\al,\be,l), \\
&& \chi_j(\al,\be,l)=(b^2-k_j^2)^{-\be}\phi_j(\al,\be,l), ~~ \varpi_j(\al,\be,l)=(a^2-k_j^2)^{-\al}\phi_j(\al,\be,l).
\end{eqnarray}
\end{subequations}
with $j=1,2,\ldots,N$.

From the previous subsections, one knows that $f_0=|(\wh{N-1})_{0,0}|$, $g_0=|(\wh{N-2},N)_{0,0}|$,
$h_0=a^N|(\wh{N-1})_{-1,0}|$, $s_0=b^N|(\wh{N-1})_{0,-1}|$ and
$\tta_0=-\frac{1}{a+b}((ab)^N|(\wh{N-1})_{-1,-1}|-|(\wh{N-1})_{0,0}|)$
solve the bilinear equations \eqref{lpKdV-Bili}, \eqref{lpmKdV-Bili-1}, \eqref{lpmKdV-Bili-2}
and \eqref{Sab-vawb-bi}, where the basic column vector $\Ph(\al,\be,l)$ satisfies system \eqref{eq:ex-CES}
with $c=0$. In the case of \eqref{ABCD}, by direct calculation one knows that  (cf. Ref. \cite{H-Bou})
\begin{align}
f_c=f_0,~g_c=g_0+Nc f_0,~h_c=h_0,~s_c=s_0,~\tta_c=\tta_0,
\label{fg0-fgdelta}
\end{align}
where in $f_c,~g_c,~h_c,~s_c$ and $\tta_c$,
the basic column vector $\Ph(\al,\be,l)$
satisfies the extended condition equation set \eqref{eq:ex-CES}.
Substituting \eqref{fg0-fgdelta} into bilinear forms \eqref{lpKdV-Bili}, \eqref{lpmKdV-Bili-1}, \eqref{lpmKdV-Bili-2}
and \eqref{Sab-vawb-bi}, one can easily find that
$f_c$, $g_c$, $h_c$, $s_c$ and $\tta_c$ also satisfy these bilinear equations. Therefore, \eqref{fghstta} together with
\eqref{phi-ce} also yields multi-soliton solutions for the lattice KdV-type equations. Furthermore, we can get the following result.
\begin{prop}
\label{Pro-lpKdV-type-CS}
The Casoratians
\begin{eqnarray}
&& f=|(\wh{N-1})_{0,0}|,~~g=|(\wh{N-2},N)_{0,0}|,~~h=a^N|(\wh{N-1})_{-1,0}|,~~s=b^N|(\wh{N-1})_{0,-1}|, \nn \\
&& \tta=-\frac{1}{a+b}((ab)^N|(\wh{N-1})_{-1,-1}|-|(\wh{N-1})_{0,0}|)
\label{fghstta-ex}
\end{eqnarray}
solve the bilinear equations \eqref{lpKdV-Bili}, \eqref{lpmKdV-Bili-1},
\eqref{lpmKdV-Bili-2} and \eqref{Sab-vawb-bi}, where the basic column vector $\Ph(\al,\be,l)$ satisfies the
extended condition equation set \eqref{eq:ex-CES}.
\end{prop}

Analogous to the earlier proofs of Propositions \ref{Pro-lpKdV-CS}-\ref{Pro-NQC-CS}, we can deduce the verification of Proposition \ref{Pro-lpKdV-type-CS}.

Constant $c$ in \eqref{eq:ex-CES} guarantees that one can get rational solutions for the lattice KdV-type equations. In fact,
$\bA_{[m]}$, $\bB_{[n]}$, $\bC_{[\be]}$ and $\bD_{[\al]}$ are taken, respectively, as
lower triangular Toeplitz matrices (For more properties of this type matrices, one can refer to Ref. \cite{ZDJ-Wronskian})
\begin{subequations}
\begin{eqnarray}
&& \bA_{[m]}=(\ga_{s,j}(q,m))_{N\times N},~~\ga_{s,j}(q,m)=\Biggl\{
\begin{array}{ll}
\frac{1}{(2s-2j)!}\partial^{2(s-j)}_k(q^2-k^2)^m\big|_{k=0},&~s\geq j,\\
0,&~s<j,
\end{array} \\
&& \bB_{[n]}=(\ga_{s,j}(p,n))_{N\times N},~~\ga_{s,j}(p,n)=\Biggl\{
\begin{array}{ll}
\frac{1}{(2s-2j)!}\partial^{2(s-j)}_k(p^2-k^2)^n\big|_{k=0},&~s\geq j,\\
0,&~s<j,
\end{array} \\
&& \bC_{[\be]}=(\ga_{s,j}(b,\be))_{N\times N},~~\ga_{s,j}(b,\be)=\Biggl\{
\begin{array}{ll}
\frac{1}{(2s-2j)!}\partial^{2(s-j)}_k(b^2-k^2)^{\be}\big|_{k=0},&~s\geq j,\\
0,&~s<j,
\end{array} \\
&& \bD_{[\al]}=(\ga_{s,j}(a,\al))_{N\times N},~~\ga_{s,j}(a,\al)=\Biggl\{
\begin{array}{ll}
\frac{1}{(2s-2j)!}\partial^{2(s-j)}_k(a^2-k^2)^{\al}\big|_{k=0},&~s\geq j,\\
0,&~s<j.
\end{array}
\end{eqnarray}
\end{subequations}
The generic basic Casoratian column vector $\Ph(\al,\be,l)$ for \eqref{eq:ex-CES} can then be
taken as
\begin{subequations}
\label{Ph-RS}
\begin{eqnarray}
\Ph(\al,\be,l)&=&\mathcal{A}_+\Ph^+(\al,\be,l)+\mathcal{A}_-\Ph^-(\al,\be,l)
\end{eqnarray}
with
\begin{eqnarray}
\Ph^{\pm}(\al,\be,l)&=&(\phi^{\pm}_0(\al,\be,l),\phi^{\pm}_1(\al,\be,l),\cdots,\phi^{\pm}_{N-1}(\al,\be,l))^T, \\
\phi^\pm_s(\al,\be,l)&=&\frac{1}{(2s)!}\partial_{k}^{2s}[(p\pm k)^n(q\pm k)^m(a\pm k)^{\al}(b\pm k)^{\be}(c\pm k)^{l+\frac{1}{2}}]\big|_{k=0},
\label{phi-1/2}
\end{eqnarray}
\end{subequations}
where $\mathcal{A}_{\pm}$ are two arbitrary non-singular lower
triangular Toeplitz matrices. In \eqref{phi-1/2}
we added $\frac{1}{2}$ in the factor $(c\pm k)^{l+\frac{1}{2}}$
to avoid zero derivative.

It is easy to understand that $S(a,b)$ given by
\eqref{NQC-tran} with \eqref{fghstta-ex} and \eqref{Ph-RS}
satisfies symmetric property $S(a,b)=S(b,a)$.
When $N=2$, solutions
$w,~v_a$ and $S(a,b)$ are, respectively, given by
\begin{subequations}
\begin{eqnarray}
w&=&\frac{2cpq}{pq+2c(mp+nq)}+2c, \\
v_a&=&\frac{-2cpq}{a(pq+2c(mp+nq))}+1, \\
S(a,b)&=&\frac{2cpq}{ab(pq+2c(mp+nq))},\label{wvS-N=2-Sab}
\end{eqnarray}
\label{wvS-N=2}
\end{subequations}
and when $N=3$, solutions $w,~v_a$ and $S(a,b)$ read, respectively,
\begin{subequations}
\begin{eqnarray}
&& w=6cpq\big(pq+2c(mp+nq)\big)^2 \nn \\
&&~~~~ /\big(-3pq((pq)^2-2cpq(mp+nq)-4c^2(mp+nq)^2)\nn \\
&&~~~~ +8c^3(3nmpq(mp+nq)+(n^3-n)q^3+(m^3-m)p^3)\big)+3c, \\
&& v_a=\frac{6cpq}{a^2}\big(pq+2c(mp+nq)\big)\big(2cpq-a(pq+2c(mp+nq))\big)\nn\\
&&~~~~ /\big(-3pq((pq)^2-2cpq(mp+nq)-4c^2(mp+nq)^2)\nn \\
&&~~~~ +8c^3(3nmpq(mp+nq)+(n^3-n)q^3+(m^3-m)p^3)\big)+1,\\
&& S(a,b)=\frac{6cpq}{a^2b^2}\big(-2cpq+a(pq+2c(mp+nq))\big)\big(-2cpq+
   b(pq+2c(mp+nq))\big)\nn \\
&&~~~~ /\big(-3pq((pq)^2-2cpq(mp+nq)-4c^2(mp+nq)^2)\nn \\
&&~~~~ +8c^3(3nmpq(mp+nq)+(n^3-n)q^3+(m^3-m)p^3)\big).\label{wvS-N=3-Sab}
\end{eqnarray}
\label{wvS-N=3}
\end{subequations}
It is noteworthy that \eqref{wvS-N=2-Sab} and \eqref{wvS-N=3-Sab} also hold for $b=-a$.

%%%%%%%%%%%%%%%%%%%%%%%%%%%%%%%%%%%%%%%%%%%%%%%%%%%%%%%%%%%%%%%%%%%%%%%%%%%%%%%%%%%%%%%%%%%%%%%%%%%%%%%%%%%%%%%%%%%%%%%%%%%%

\section{Rational solutions for the ABS list $\rm{Q3}_{\del}$}

\subsection{Rational solutions for the $\rm{Q3}_{\del}$}

In Ref. \cite{Nijhoff-CM-2009}, soliton solutions for $\rm{Q3_{\del}}$ \eqref{eq:Q3parm} were
written as a linear combination of four terms each of which contains as an essential
ingredient the soliton solution of NQC equation with different values of
the branch point parameters which enter in that equation. In the following, we still
adopt this result to present the rational solution for the $\rm{Q3}_{\del}$.

\begin{Theorem} \label{theorem-1}
The rational solution of ${\rm Q3}_{\del}$ \eqref{eq:Q3parm} is formulated by
\begin{equation}
\label{eq:Q3sol}
\begin{split}
u =& A\digamma(a,b)\left[ 1-(a+b)S(a,b)\right]+B\digamma(a,-b)\left[1-(a-b)S(a,-b)\right]\\
& + C\digamma(-a,b)\left[ 1+(a-b)S(-a,b)\right]+ D\digamma(-a,-b)\left[ 1+(a+b)S(-a,-b)\right],
\end{split}
\end{equation}
in which $S(\pm a,\pm b)$ are the rational solutions of the NQC equation \eqref{eq:NQC} with parameters $\pm a, \pm b$;
the function $\digamma(a,b)$ is defined as
\begin{equation}
\digamma(a,b) = \left(\dfrac{P}{(p-a)(p-b)}\right)^n\left(\dfrac{Q}{(q-a)(q-b)}\right)^m,
\label{eq:vpdef}
\end{equation}
and $P,Q$ are defined by \eqref{elliptic-Q3}; $A$, $B$, $C$ and $D$ are
constants subject to the single constraint
\begin{equation}\label{eq:ABCD}
AD(a+b)^2-BC(a-b)^2=-\frac{\del^2}{16ab}.
\end{equation}
\end{Theorem}

The proof of soliton solutions to the $\rm{Q3}_{\del}$ presented in Ref. \cite{Nijhoff-CM-2009}
is based on lpKdV equation \eqref{lpKdV}, Miura transformations \eqref{lpKdV-lpmKdV-MT} and \eqref{Sab-vawb}, as well as symmetric property
$S(a,b)=S(b,a)$. Since rational solutions for these equations have been shown in previous section,
naturally we achieve the verification of Theorem \ref{theorem-1}. We omit it here.

\subsection{Degeneration}\label{sec:ABS}

We now consider the problem of degeneration of rational solutions into the remaining ``lower'' equations
$\rm{Q2, Q1_{\del}, H3_{\del}, H2}$ and $\rm{H1}$ in the ABS list \eqref{ABS-list-p}. To do so we follow the degenerations given in Ref. \cite{Nijhoff-CM-2009} which
are limits on the parameters $a$ and $b$ and the dependent variable $u$, where a small parameter
$\epsilon$ is introduced, and all degenerations are obtained in the limit $\epsilon\rightarrow 0$.
The degeneration relations between ${\rm Q3_{\del}}$ and ``lower equations'' ${\rm Q2, Q1_{\del}, H3_{\del}, H2}$ and ${\rm H1}$
can be depicted by Fig.1.
\vspace{-5mm}
\begin{center}
\begin{displaymath}
\xymatrix{ \boxed{\mathrm{Q3_{\del}}} \ar[d] \ar[r] & \boxed{\mathrm{Q2}} \ar[d] \ar[r] & \boxed{\mathrm{Q1_{\del}}} \ar[d] \\
\boxed{\mathrm{H3_{\del}}} \ar[r] & \boxed{\mathrm{H2}} \ar[r] & \boxed{\mathrm{H1}} }
\end{displaymath}
\begin{minipage}{11cm}{\footnotesize~~~~~~~~~~~~~~~~~~~~~~~~~~~~~~~~~
{Fig.1} Degeneration relation}
\end{minipage}
\end{center}

\subsubsection{${\rm Q3_{\del}} \longrightarrow {\rm Q2}$}

The degeneration from ${\rm Q3_{\del}}$ to ${\rm Q2}$ is
\begin{eqnarray}
b=a(1-2\epsilon),~~ u\longrightarrow \dfrac{\delta}{4a^2}\left(\dfrac{1}{\epsilon}+1+(1+2u)\epsilon\right).
\end{eqnarray}
Making the following replacements of constants in \eqref{eq:Q3sol}
\begin{equation}
\begin{array}{l}
A \rightarrow \dfrac{\delta}{4a^2}A\epsilon,~~B \rightarrow \dfrac{\delta}{8a^2}\left(\dfrac{1}{\epsilon}+1-\xi_0+((3+\xi_0^2)/2+2AD)\epsilon\right),\\
C \rightarrow \dfrac{\delta}{8a^2}\left(\dfrac{1}{\epsilon}+1+\xi_0+((3+\xi_0^2)/2+2AD)\epsilon\right),~~
D \rightarrow \dfrac{\delta}{4a^2}D\epsilon,
\end{array}
\label{eq:Q2const}
\end{equation}
we find the rational solution for ${\rm Q2}$:
\begin{equation}
\begin{split}
u =& \dfrac{1}{4}((\xi+\xi_0)^2+1)+a(\xi+\xi_0)\UM(-a,a)+a^2\left(\UU(a,-a)+\UU(-a,a)\right) + \\
& AD + \dfrac{1}{2}A\rho(a)(1-2a\UM(a,a))+\dfrac{1}{2}D\rho(-a)(1+2a\UM(-a,-a)),
\end{split}
\label{eq:Q2sol}
\end{equation}
in which
\begin{equation}
\begin{array}{rl}
\xi=& 2a\left(\dfrac{p}{a^2-p^2}n+\dfrac{q}{a^2-q^2}m\right), ~~\rho(a)=\left(\frac{p+a}{p-a}\right)^n\left(\frac{q+a}{q-a}\right)^m, \\
\UU(a,-a)=& -\frac{1}{2a}\partial_{\epsilon}S(a,2a\epsilon-a)\big|_{\epsilon=0}, ~~
\UU(-a,a) = \frac{1}{2a}\partial_{\epsilon}S(-a,a-2a\epsilon)\big|_{\epsilon=0}, \\
\end{array}
\label{eq:xietaUU-1}
\end{equation}
and $\xi_0, A$ and $D$ are the constants which may be chosen arbitrarily.

\subsubsection{${\rm Q2} \longrightarrow {\rm Q1_{\del}}$}

To achieve the rational solution for ${\rm Q1_{\del}}$ we degenerate from \eqref{eq:Q2sol} by taking
\[u\longrightarrow\dfrac{\delta^2}{4\epsilon^2}+\dfrac{1}{\epsilon}u.\]
Meanwhile, we replace the constants appearing in solution \eqref{eq:Q2sol} by
\begin{equation}
A \rightarrow \dfrac{2A}{\epsilon},~~
D \rightarrow \dfrac{2D}{\epsilon},~~
\xi_0 \rightarrow \xi_0 + \dfrac{2B}{\epsilon}.
\label{eq:Q1const}
\end{equation}
Then the rational solutions for ${\rm Q1_{\del}}$ can be described as
\begin{equation}
u=A\rho(a)(1-2a\UM(a,a)) + B(\xi+\xi_0+2a\UM(-a,a)) + D\rho(-a)(1+2a\UM(-a,-a)),
\label{eq:Q1sol}
\end{equation}
where constants
$A$, $B$, $D$ and $\xi_0$ are chosen to satisfy the single constraint
\begin{equation}
AD +\dfrac{1}{4}B^2=\dfrac{\delta^2}{16}.
\end{equation}

%The rational solutions of Q1 \eqref{Q1-p} which emerges is

\subsubsection{${\rm Q3_{\del}} \longrightarrow {\rm H3_{\del}}$}

By setting
\begin{equation}
b=\dfrac{1}{\epsilon^2},~~u\longrightarrow \epsilon^3\dfrac{\sqrt{\delta}}{2}u,
\end{equation}
and
\begin{equation}
A\rightarrow \epsilon^3\dfrac{\sqrt{\delta}}{2}A,~~
B\rightarrow \epsilon^3\dfrac{\sqrt{\delta}}{2}B,~~
C\rightarrow \epsilon^3\dfrac{\sqrt{\delta}}{2}C,~~
D\rightarrow \epsilon^3\dfrac{\sqrt{\delta}}{2}D,
\label{eq:H3const}
\end{equation}
rational solution to ${\rm H3_{\del}}$ can be degenerated from \eqref{eq:Q3sol}, which is of form
\begin{equation}\label{eq:H3sol}
u=(A+(-1)^{n+m}B)\vrr v_a+((-1)^{n+m}C+D)\vrr^{-1}v_{-a},
\end{equation}
in which $v_a$ is defined by \eqref{va} and
\begin{eqnarray}
\vrr=\bigg(\frac{P}{a-p}\bigg)^n\bigg(\frac{Q}{a-q}\bigg)^m,
\end{eqnarray}
where parameters $P$ and $Q$ are related to $p$ and $q$ by \eqref{eq:parcurves} and the constants $A$, $B$, $C$ and $D$ are subject to the constraint
$$ AD-BC = \dfrac{\delta}{4a}.$$

\subsubsection{${\rm Q2} \longrightarrow {\rm H2}$}

The degeneration from ${\rm Q2}$ to ${\rm H2}$ can be arrived at by setting
\begin{equation}
a=\dfrac{1}{\epsilon},~~u\longrightarrow\frac{1}{4}+\epsilon^2u.
\label{eq:U-exp-3}
\end{equation}
Substituting \eqref{eq:U-exp-3} into \eqref{eq:Q2sol}
combined with
\begin{equation}
\begin{array}{rl}
a\UM(-a,a) \longrightarrow & -\epsilon S^{(0)} + O(\epsilon^2),\\
a\UM(a,a) \longrightarrow & \epsilon S^{(0)}-2\epsilon^2 S^{(1)} + O(\epsilon^3),\\
a\UM(-a,-a) \longrightarrow & \epsilon S^{(0)}+2\epsilon^2 S^{(1)} + O(\epsilon^3),\\
a^2(\UU(-a,a)+\UU(a,-a)) \longrightarrow & 2\epsilon^2 S^{(1)} + O(\epsilon^3),
\end{array}
\label{eq:U-exp-33}
\end{equation}
and the following choice for the constants
\begin{equation}
A\rightarrow A(\epsilon+\zeta_1\epsilon^2/2),~~
D\rightarrow A(-\epsilon+\zeta_1\epsilon^2/2),~~
\xi_0\rightarrow \epsilon \zeta_0
\end{equation}
with unconstrained constants $\zeta_0,\zeta_1$, the rational solution for ${\rm H2}$ reads
\begin{equation}\label{eq:H2sol}
u=\dfrac{1}{4}(\zeta+\zeta_0)^2 -(\zeta+\zeta_0) S^{(0)}+2S^{(1)}-A^2+(-1)^{n+m}A(\zeta+\zeta_1/2-2S^{(0)}),
\end{equation}
where
\begin{equation}
\zeta=2(np+mq).
\label{eq:zdef}
\end{equation}

From \eqref{wvS-N=2} and \eqref{wvS-N=3} one know that when $N=2$, $S^{(0)}$ and $S^{(1)}$ are taken as
\begin{eqnarray}
S^{(0)}=\frac{2cpq}{pq+2c(mp+nq)}, ~~S^{(1)}=0,
\end{eqnarray}
and when $N=3$, $S^{(0)}$ and $S^{(1)}$ are described as
\begin{eqnarray*}
&& S^{(0)}=6cpq(pq+2c(mp+nq))^2/\big(-3pq((pq)^2-2cpq(mp+nq)-4c^2(mp+nq)^2)\nn \\
&&~~~~ +8c^3(3nmpq(mp+nq)+(n^3-n)q^3+(m^3-m)p^3)\big), \\
&& S^{(1)}=12(cpq)^2(pq+2c(mp+nq))/\big(-3pq((pq)^2-2cpq(mp+nq)-4c^2(mp+nq)^2)\nn \\
&&~~~~ +8c^3(3nmpq(mp+nq)+(n^3-n)q^3+(m^3-m)p^3)\big).
\end{eqnarray*}

\subsubsection{${\rm Q1}_{\del} \longrightarrow {\rm H1}$}

The rational solution to ${\rm H1}$ can be obtained from \eqref{eq:Q1sol} through degeneration.
Substituting
\begin{equation}
a=\dfrac{1}{\epsilon},~~u\rightarrow \epsilon \delta u
\end{equation}
into \eqref{eq:Q1sol} and using
\begin{equation}
A\rightarrow  \dfrac{\delta}{2} A(1+\zeta_1\epsilon),~~
D\rightarrow  \dfrac{\delta}{2} A(-1+\zeta_1\epsilon),~~
B\rightarrow  \delta B,~~
\xi_0 \rightarrow  \epsilon \zeta_0
\end{equation}
with constants $\zeta_0,\zeta_1$, we have
\begin{equation}\label{eq:H1sol}
u=B(\zeta+\zeta_0-2S^{(0)})+(-1)^{n+m}A(\zeta+\zeta_1-2S^{(0)}),
\end{equation}
which presents rational solution of H1,
where $\zeta_0$, $\zeta_1$, $A$ and $B$ satisfies
$ A^2-B^2=-\frac{1}{4}$. Here $S^{(0)}$ is same as the one in solution \eqref{eq:H2sol}.

%%%%%%%%%%%%%%%%%%%%%%%%%%%%%%%%%%%%%%%%%%%%%%%%%%%%%%%%%%%%%%%%%%%%%%%%%%%%%%%%%%%%%%%%%%%%%%%%%%%%

\section{Conclusions}

In the present paper, we investigate the rational solutions for the whole ABS list except for
${\rm Q4}$. The procedure is different from the one given in Ref. \cite{ZZ-ABS-RS}.
We make use of Hirota's bilinear method together with the key results given in Ref. \cite{Nijhoff-CM-2009}.
In order to express the rational solutions for the lattice KdV-type equations uniformly, we consider an
extended condition equation set \eqref{eq:ex-CES}, where a constant $c$ is introduced. This constant
is indispensable and allows one to take spectral parameters' limit $k\rightarrow 0$.
On basis of rational solutions for the lpKdV equation and the Miura transformations \eqref{lpKdV-lpmKdV-MT} and \eqref{Sab-vawb},
together with symmetric property $S(a,b)=S(b,a)$,
one can finish the proof of Theorem \ref{theorem-1}. Theorem \ref{theorem-1} reveals a fact that solutions
to ${\rm Q3_{\del}}$ can be expressed as a linear combination of four different solutions of the NQC equation with the parameters $a, b$ changing signs
holds not only for solitons \cite{Nijhoff-CM-2009} but also for rational solutions.
Besides, from the relation \eqref{eq:Q3sol} one may consider \eqref{Sab-vawb-bi} as a bilinear form of ${\rm Q3_{\del}}$.
Based on the degeneration relations between ${\rm Q3_{\del}}$ and ``lower equations'' ${\rm Q2, Q1_{\del}, H3_{\del}, H2}$ and ${\rm H1}$,
the rational solutions for the latter equations are also derived.

\subsection*{Acknowledgements}

This project is supported by the NSF of China (Nos. 11371241, 11631007, 11301483, 11401529), SRF of
the DPHE of China (No. 20113108110002) and the Natural Science Foundation of Zhejiang Province (No. LY17A010024).

%%%%%%%%%%%%%%%%%%%%%%%%%%%%%%%%%%%%%%%%%%%%%%%%%%%%%%%%%%%%%%%%%%%%%%%%%%%%%%%%%%%%%%%%%%%%%%%%%%%%%%%%%%%%%%%%%%%%%%%
\appendix

\section{Casoratian shift formulae}
\label{A:a}

We list shift formulae for the Casoratians
\begin{equation*}
f=|(\wh{N-1})_{0,0}|,~g=|(\wh{N-2},N)_{0,0}|,~h=|(\wh{N-1})_{-1,0}|,~s=|(\wh{N-1})_{0,-1}|,~
\vta=|(\wh{N-1})_{-1,-1}|,
\end{equation*}
where the basic column vector $\Ph(\al,\be,l)$ satisfies the relations \eqref{eq:CES-1}, \eqref{eq:CES-2}
and \eqref{eq:CES-3-4}.
\begin{subequations}
\label{eq:lpKdV-iden-1}
\begin{eqnarray}
&&  p^{N-2}\dt f=-|(\wh{N-2})_{0,0},\dt{\Ph}(0,0,N-2)|, \label{eq:lpKdV-iden-1a} \\
&&  q^{N-2}\dh f=-|(\wh{N-2})_{0,0},\dh{\Ph}(0,0,N-2)|, \label{eq:lpKdV-iden-1b} \\
&& (p-a)p^{N-2}\dt h=a^N|0_{-1,0},(\wh{N-3})_{0,0},\dt{\Ph}(0,0,N-2)|, \label{eq:lpmKdV-iden-1c} \\
&& (q-a)q^{N-2}\dh h=a^N|0_{-1,0},(\wh{N-3})_{0,0},\dh{\Ph}(0,0,N-2)|, \label{eq:lpmKdV-iden-1d} \\
&&  p^{N-2}\dt s=-b^N|(\wh{N-2})_{0,-1},\dt{\Ph}(0,-1,N-2)|, \label{eq:NQC-iden-1a} \\
&&  q^{N-2}\dh s=-b^N|(\wh{N-2})_{0,-1},\dh{\Ph}(0,-1,N-2)|, \label{eq:NQC-iden-1b} \\
&& (p-a)p^{N-2}\dt{\vta}=(ab)^N|0_{-1,-1},(\wh{N-3})_{0,-1},\dt{\Ph}(0,-1,N-2)|, \label{eq:NQC-iden-1c} \\
&& (q-a)q^{N-2}\dh{\vta}=(ab)^N|0_{-1,-1},(\wh{N-3})_{0,-1},\dh{\Ph}(0,-1,N-2)|, \label{eq:NQC-iden-1c1} \\
&& (p-b)p^{N-2}\dt{\vta}=(ab)^N|0_{-1,-1},(\wh{N-3})_{-1,0},\dt{\Ph}(-1,0,N-2)|, \label{eq:NQC-iden-1d} \\
&& (q-b)q^{N-2}\dh{\vta}=(ab)^N|0_{-1,-1},(\wh{N-3})_{-1,0},\dh{\Ph}(-1,0,N-2)|, \label{eq:NQC-iden-1d1} \\
&& p^{N-2}(\dt g+p\dt f)=-|(\wh{N-3},N-1)_{0,0},\dt{\Ph}(0,0,N-2)|, \label{eq:lpKdV-iden-1e} \\
&& q^{N-2}(\dh g+q\dh f)=-|(\wh{N-3},N-1)_{0,0},\dh{\Ph}(0,0,N-2)|, \label{eq:lpKdV-iden-1f} \\
&& (p-q)p^{N-2}q^{N-2}\dh{\dt{f}}=|(\wh{N-3})_{0,0},\dh{\Ph}(0,0,N-2),\dt{\Ph}(0,0,N-2)|, \label{eq:lpKdV-iden-1g}
\end{eqnarray}
\end{subequations}
and
\begin{subequations}
\label{eq:lpKdV-iden-2}
\begin{eqnarray}
&& p^{N-2}\wt{f}=\frac{|\wt{\bB}_{[n]}|}{|\bB_{[n]}|}|(\wh{N-2})_{0,0},
\bB_{[n]}\wt{\bB}_{[n]}^{-1}\wt{\Ph}(0,0,N-2)|,\label{eq:lpKdV-iden-2a}\\
&& q^{N-2}\wh{f}=\frac{|\wh{\bA}_{[m]}|}{|\bA_{[m]}|}|(\wh{N-2})_{0,0},
\bA_{[m]}\wh{\bA}_{[m]}^{-1}\wh{\Ph}(0,0,N-2)|,\label{eq:lpKdV-iden-2b}\\
&& b^{N-2}\d{s}=b^N\frac{|\d{\bC}_{[\be]}|}{|\bC_{[\be]}|}|(\wh{N-2})_{0,-1},
\bC_{[\be]}\d{\bC}_{[\be]}^{-1}\d{\Ph}(0,-1,N-2)|,\label{eq:NQC-iden-2a}\\
&& a^{N-2}\cd{h}=a^N\frac{|\cd{\bD}_{[\al]}|}{|\bD_{[\al]}|}|(\wh{N-2})_{-1,0},
\bD_{[\al]}\cd{\bD}_{[\al]}^{-1}\cd{\Ph}(-1,0,N-2)|,\label{eq:lpKdV-iden-2b1}\\
&& (p+a)p^{N-2}\wt{h}=a^N\frac{|\wt{\bB}_{[n]}|}{|\bB_{[n]}|}
|0_{-1,0},(\wh{N-3})_{0,0},\bB_{[n]}\wt{\bB}_{[n]}^{-1}\wt{\Ph}(0,0,N-2)|,\label{eq:lpKdV-iden-2c}\\
&& (q+a)q^{N-2}\wh{h}=a^N\frac{|\wh{\bA}_{[m]}|}{|\bA_{[m]}|}
|0_{-1,0},(\wh{N-3})_{0,0},\bA_{[m]}\wh{\bA}_{[m]}^{-1}\wh{\Ph}(0,0,N-2)|, \label{eq:lpKdV-iden-2d}\\
&& (p+q)p^{N-2}q^{N-2}\wh{\dt{f}}=\frac{|\wh{\bA}_{[m]}|}{|\bA_{[m]}|}|(\wh{N-3})_{0,0},\dt{\Ph}(0,0,N-2),
\bA_{[m]}\wh{\bA}_{[m]}^{-1}\wh{\Ph}(0,0,N-2)|,\label{eq:lpKdV-iden-2e}\\
&& (p+q)p^{N-2}q^{N-2}\wt{\dh{f}}=\frac{|\wt{\bB}_{[n]}|}{|\bB_{[n]}|}|(\wh{N-3})_{0,0},\dh{\Ph}(0,0,N-2),
\bB_{[n]}\wt{\bB}_{[n]}^{-1}\wt{\Ph}(0,0,N-2)|,\label{eq:lpKdV-iden-2f}\\
&& (p+b)p^{N-2}b^{N-2}\d{\dt{s}}=b^N\frac{|\d{\bC}_{[\be]}|}{|\bC_{[\be]}|}
|(\wh{N-3})_{0,-1},\dt{\Ph}(0,-1,N-2),\bC_{[\be]}\d{\bC}_{[\be]}^{-1}\d{\Ph}(0,-1,N-2)|,\label{eq:NQC-iden-2b}\\
&& (q+b)q^{N-2}b^{N-2}\d{\dh{s}}=b^N\frac{|\d{\bC}_{[\be]}|}{|\bC_{[\be]}|}
|(\wh{N-3})_{0,-1},\dh{\Ph}(0,-1,N-2),\bC_{[\be]}\d{\bC}_{[\be]}^{-1}\d{\Ph}(0,-1,N-2)|,\label{eq:NQC-iden-2c}\\
&& (a+b)b^{N-2}\d{\dc{s}}=b^N\frac{|\d{\bC}_{[\be]}|}{|\bC_{[\be]}|}
|0_{-1,-1},(\wh{N-3})_{0,-1},\bC_{[\be]}\d{\bC}_{[\be]}^{-1}\d{\Ph}(0,-1,N-2)|,\label{eq:NQC-iden-2d}\\
&& (p+a)p^{N-2}a^{N-2}\cd{\dt{h}}=a^N\frac{|\cd{\bD}_{[\al]}|}{|\bD_{[\al]}|}
|(\wh{N-3})_{-1,0},\dt{\Ph}(-1,0,N-2),\bD_{[\al]}\cd{\bD}_{[\al]}^{-1}\cd{\Ph}(-1,0,N-2)|,\label{eq:NQC-iden-2b1}\\
&& (q+a)q^{N-2}a^{N-2}\cd{\dh{h}}=a^N\frac{|\cd{\bD}_{[\al]}|}{|\bD_{[\al]}|}
|(\wh{N-3})_{-1,0},\dh{\Ph}(-1,0,N-2),\bD_{[\al]}\cd{\bD}_{[\al]}^{-1}\cd{\Ph}(-1,0,N-2)|,\label{eq:NQC-iden-2c1}\\
&& (a+b)a^{N-2}\cd{\dd{h}}=a^N\frac{|\cd{\bD}_{[\al]}|}{|\bD_{[\al]}|}
|0_{-1,-1},(\wh{N-3})_{-1,0},\bD_{[\al]}\cd{\bD}_{[\al]}^{-1}\cd{\Ph}(-1,0,N-2)|,\label{eq:NQC-iden-2d1}\\
&& q^{N-2}(\wh{g}-q\wh{f})=
\frac{|\wh{\bA}_{[m]}|}{|\bA_{[m]}|}|(\wh{N-3},N-1)_{0,0},\bA_{[m]}\wh{\bA}_{[m]}^{-1}\wh{\Ph}(0,0,N-2)|.\label{eq:lpKdV-iden-2g}
\end{eqnarray}
\end{subequations}

%%%%%%%%%%%%%%%%%%%%%%%%%%%%%%%%%%%%%%%%%%%%%%%%%%%%%%%%%%%%%%%%%%%%%%%%%%%%%%%%%%%%%%%%%%%%%%%%%%%%%%%%%%%%%%%%%%%%%%%

%%%%%%%%%%%%%%%%%%%%%%%%%%%%%%%%%%%%%%%%%%%%%%%%%%%%%%%%%%%%%%%%%%%%%%%%%%%%%%%%%%%%%%%%%%%%%%%%%%%%%%%%%%%%%%%%%%%%%%%

\end{document}